\documentclass[twoside,leqno,twocolumn]{article}
\usepackage[letterpaper]{geometry}
\usepackage{ltexpprt}
\usepackage{xcolor}
\usepackage[colorlinks]{hyperref}
\definecolor{lightblue}{rgb}{0.5,0.5,1.0}
\definecolor{darkred}{rgb}{0.5,0,0}
\definecolor{darkgreen}{rgb}{0,0.5,0}
\definecolor{darkblue}{rgb}{0,0,0.5}
\hypersetup{colorlinks,linkcolor=darkred,filecolor=darkgreen,urlcolor=darkred,citecolor=darkblue}
\usepackage{tabularx}
\usepackage{braket}
\usepackage{cite}
\usepackage{nicefrac}
\usepackage{amssymb}
\usepackage{mathtools}

\usepackage[linesnumbered,ruled,commentsnumbered,longend,algo2e]{algorithm2e}

\usepackage[edges]{forest}
\usetikzlibrary{shapes,arrows,positioning,calc,decorations.pathmorphing}

\DeclareMathOperator{\id}{id}
\DeclareMathOperator{\Aut}{Aut}
\DeclareMathOperator{\Refx}{Ref}
\let\Sel\undefined
\DeclareMathOperator{\Sel}{Sel}
\DeclareMathOperator{\Inv}{Inv}

\DeclarePairedDelimiter{\ceil}{\lceil}{\rceil}

\definecolor{blue0}{RGB}{196, 196, 255}
\definecolor{blue1}{RGB}{140, 140, 255}
\definecolor{blue2}{RGB}{112, 112, 255}
\definecolor{blue3}{RGB}{0, 0, 255}

\usetikzlibrary{patterns}

\tikzset{
	hatch distance/.store in=\hatchdistance,
	hatch distance=10pt,
	hatch thickness/.store in=\hatchthickness,
	hatch thickness=2pt
}

\makeatletter
\pgfdeclarepatternformonly[\hatchdistance,\hatchthickness]{flexible hatch}
{\pgfqpoint{0pt}{0pt}}
{\pgfqpoint{\hatchdistance}{\hatchdistance}}
{\pgfpoint{\hatchdistance-1pt}{\hatchdistance-1pt}}%
{
	\pgfsetcolor{\tikz@pattern@color}
	\pgfsetlinewidth{\hatchthickness}
	\pgfpathmoveto{\pgfqpoint{0pt}{0pt}}
	\pgfpathlineto{\pgfqpoint{\hatchdistance}{\hatchdistance}}
	\pgfusepath{stroke}
}
\makeatother

\tikzset{pruned/.style={fill=gray,pattern=flexible hatch,hatch distance=3pt,hatch thickness=0.5pt,pattern color=lightgray}}

\tikzset{snake it/.style={decorate, decoration=snake, segment length=4mm, thick, draw=blue}}

\newcommand{\Traces}{\textsc{Traces}}
\newcommand{\nauty}{\textsc{nauty}}
\newcommand{\bliss}{\textsc{bliss}}

\newcommand{\dejavu}{\textsc{dejavu}}
\newcommand{\conauto}{\textsc{conauto}}

\DeclareMathOperator{\Dev}{Dev}

\SetKwProg{Fn}{function}{}{end}\SetKwFunction{RandomAut}{Automorphisms}
\SetKwProg{Fn}{function}{}{end}\SetKwFunction{RandomIso}{Isomorphism}
\SetKwFunction{FirstLeaf}{FirstLeaf}
\SetKwFunction{RandomWalk}{RandomWalk}
\SetKwFunction{NextNodeDFS}{OneNodeDFS}
\SetKwFunction{NextNodeBFS}{OneNodeBFS}
\SetKwFunction{Some}{Some}
\SetKwFunction{Sift}{Sift}
\SetKwFunction{RandomElement}{RandomElement}
\SetKwFunction{Refine}{Refine}
\SetKwFunction{RRef}{Ref}
\SetKwFunction{SSel}{Sel}
\SetKwFunction{Sift}{Sift}
\SetKwFunction{CertifyAutomorphism}{CertifyAutomorphism}
\SetKwFunction{CertifyIsomorphism}{CertifyIsomorphism}



\title{Engineering a Fast Probabilistic Isomorphism Test\thanks{The research leading to these results has received funding from the European Research Council (ERC) under the European Union’s Horizon 2020 research and innovation programme (EngageS: grant agreement No.\ 820148).}}
\author{Markus Anders \and Pascal Schweitzer}
\date{}
\usepackage[final]{pdfpages}

\begin{document}

\maketitle

\begin{abstract} 
We engineer a new probabilistic Monte-Carlo algorithm for isomorphism testing. Most notably, as opposed to all other solvers, it implicitly exploits the presence of symmetries without explicitly computing them.

We provide extensive benchmarks, showing that the algorithm outperforms all state-of-the-art solutions for isomorphism testing on most inputs from the de facto standard benchmark library for isomorphism testing. On many input types, our data not only show improved running times by an order of magnitude, but also reflect a better asymptotic behavior. 

Our results demonstrate that, with current algorithms, isomorphism testing is in practice easier than the related problems of computing the automorphism group or canonically labeling a graph. The results also show that probabilistic algorithms for isomorphism testing can be engineered to outperform deterministic approaches, even asymptotically.
\end{abstract}

\cfoot{\thepage}

\section{Motivation} \label{sec:introduction}
The \emph{graph isomorphism problem} is concerned with deciding whether two given graphs are structurally equivalent. 
It captures the essence of symmetry detection in combinatorial structures. 
Two different, strongly related problems are commonly considered in practice: first, the \emph{automorphism group} problem demands computation of the entire symmetry group of a given graph. 
Secondly, the problem of computing a \emph{canonical labeling} asks us to produce an ordering of the vertices of a given input graph, so that isomorphic inputs yield equivalent ordered graphs.

From a theoretical point of view the graph isomorphism and automorphism group problems are polynomial time equivalent (see~\cite{DBLP:journals/ipl/Mathon79}). To reduce the isomorphism to the automorphism problem, one essentially computes the automorphism group of the disjoint union of the two input graphs.
The other direction is a Turing reduction, however, so to compute the automorphism group one requires repeated isomorphism tests of suitably manipulated graphs. 
Both problems reduce to the task of computing canonical labelings.
Many theoretical isomorphism testing algorithms, albeit sometimes with considerable extra effort, can be extended or modified to produce canonical labelings. For example the best known theoretical isomorphism algorithm due to Babai \cite{DBLP:conf/stoc/Babai16} can be extended to produce canonical labelings~\cite{DBLP:conf/stoc/Babai19}. 
However, the canonical labeling problem is not known to be polynomial-time reducible to the other two problems and thus it is the potentially harder problem.

In practice, the situation of the relationship between the problems is slightly different. Currently, all state-of-the-art tools are based on the so-called \emph{individualization-refinement} (IR) paradigm. 
The paradigm performs a form of backtracking leading to a \emph{search tree}, which can be of exponential size in the original graph. 
Algorithms then traverse and prune the search tree in certain manners to solve the three problems related to graph isomorphism mentioned above. 
Despite the fact that the practical tools mostly traverse the same search tree, there is diversity among them. 
Differences between practical tools manifest in the choice of traversal strategies, pruning techniques and various optimization tricks. 

With the exception of \conauto{} \cite{DBLP:conf/wea/Lopez-PresaCA13}, modern practical tools have no specific mode for isomorphism testing, hence either computing entire automorphism groups or canonical labelings \cite{JunttilaKaski:ALENEX2007, DBLP:conf/tapas/JunttilaK11,McKay201494,Darga:2004:ESS:996566.996712,conauto:webpage}.
In their implementation, it is often the case that computing canonical labelings is significantly more expensive than automorphism group computation, since there are fewer known algorithmic techniques and tricks that can be applied (see \cite{DBLP:conf/birthday/KatebiSM12,McKay201494}). 
However, as described in \cite{McKay:userguide},
the fastest practical way to decide the isomorphism problem for most inputs currently comprises first in computing canonical labelings for both input graphs and then comparing the outputs. 
Of course we can alternatively employ the aforementioned reduction from isomorphism testing to automorphism group computation by computing the automorphism group on the disjoint union of graphs. 
One might hope that this leads to faster computation since we can use more algorithmic techniques and tricks as just mentioned.
But actually computing the automorphism groups of such disjoint unions of graphs turns out to be even more expensive for most cases.

A crucial difference is that the isomorphism problem only requires us to know whether \emph{one} isomorphism exists, while automorphism group based algorithms produce \emph{all} isomorphisms and automorphisms. 
The entire automorphism group is commonly managed using the Schreier-Sims algorithm (see~\cite{MR1970241}).
Essentially, automorphisms are collected in a table, which can grow to a quadratic size in the original graph.
The potential quadratic blowup for collecting automorphisms is however not limited to automorphism group computation but also happens when computing canonical labelings. Indeed, the solvers rely on finding the entire automorphism group to then in turn prune search for the canonical form of the graph.   
In fact, even the state-of-the-art isomorphism test mentioned earlier \cite{DBLP:conf/wea/Lopez-PresaCA13} collects automorphisms when testing for isomorphism. 

Overall, it seems that the current practical reality is, in a sense, upside down: testing for isomorphism should be easier than providing automorphism groups or canonical labeling, but turns out to be the computationally most expensive task -- inadvertently solving the other two supposedly ``harder'' problems in the process.

\textbf{Contribution.} In this paper, we engineer a new probabilistic Monte-Carlo algorithm for isomorphism testing. 
Most notably, it implicitly exploits the presence of automorphisms without explicitly computing automorphism groups. In particular we can avoid the use of the Schreier-Sims algorithm. Some aspects of our algorithm are related to an older randomized approach for isomorphism testing~\cite{DBLP:conf/alenex/KutzS07} which is however not competitive with current tools. In comparison, with our approach there is a crucial difference in that the underlying algorithm is based around the \emph{bidirectional search} traversal strategy presented in \cite{theorypaper}. It features a provably superior theoretical worst-case running time over deterministic solvers in the IR-paradigm with exponential speed up.
It turns out that this superior running time not only manifests in worst-case complexity and not only in theory.
Indeed, we provide extensive benchmarks, showing that the algorithm outperforms all state-of-the-art solutions for isomorphism testing on most inputs from the de facto standard benchmark library for isomorphism testing~\cite{nautyTracesweb}.
On many input types, our data not only show improved running times by an order of magnitude, but also reflect a better asymptotic behavior. 

The algorithm itself is solely based on repeatedly probing random walks of the in\-dividual\-ization-\-re\-fine\-ment search trees.
Since isomorphic graphs imply isomorphic search trees, random walks either produce equally probable outcomes in both trees, or entirely distinct outcomes.
Based on this, the algorithm performs a probabilistic test for deciding isomorphism.  
By memorizing the results of all previous random walks, it exploits the ideas of the well-known \emph{collision problem}, since any repetition in the results of walks advances the progress of the algorithm. 

To summarize, this paper paints a refined picture of the practical graph isomorphism landscape: firstly, we establish isomorphism testing as the easiest of the practical problems.  
This makes canonical labeling the hardest practical problem, followed by automorphism group computation and lastly isomorphism testing. 
Secondly, we show that probabilistic algorithms can be engineered to outperform deterministic approaches. This is in agreement with the  theoretical worst-cases analysis of IR-algorithms performed in~\cite{theorypaper}.

\section{Individualization-Refinement} \label{sec:preliminaries}
We follow the descriptions given in \cite{McKay201494}, giving a quick introduction into the individualization-refinement framework. 
The presentation is geared towards results necessary for the probabilistic isomorphism test presented in Section~\ref{sec:bidirectional}.

\paragraph{Graphs and groups.} An undirected, finite graph $G = (V, E)$ consists of a set of vertices $V \subseteq \mathbb{N}$ and a set of edges $E \subseteq \binom{V}{2}$. For simplicity, we let $V = \{1, \dots{}, n\}$. We denote by~$S_n$ the \emph{symmetric group} on~$\{1, \dots{}, n\}$.

A \emph{coloring} is a surjective map~$\pi : V \to \{1, \dots{}, k\}$. It maps the vertices of a graph to \emph{cells} $1, \dots{}, k$. The \emph{$i$-th cell} of $\pi$ is $\pi^{-1}(i) \subseteq V$. Consequently, $|\pi| = k$ denotes the number of cells in a given coloring. We call $\pi$ \emph{discrete} whenever $|\pi| = n$. Since a discrete coloring is then bijective, it also implicitly orders the vertices in~$V$.

A colored graph $(G, \pi)$ consists of a graph and a coloring. We require that isomorphisms and automorphisms of a colored graph must preserve colors, i.e., a vertex of a cell $c$ must be mapped to a vertex of cell $c$.

Let $G_1 = (V_1, E_1)$ and $G_2 = (V_2, E_2)$ denote two graphs. 
A bijection $\varphi : V_1 \to V_2$ is an \emph{isomorphism} whenever $G_1^\varphi := (\varphi(V_1), \varphi(E_1)) = (V_2, E_2) = G_2$ holds. 
If $G_1 = G_2$, we call $\varphi$ an \emph{automorphism} of $G_1$.
The set containing all automorphisms of a graph $G$ forms a permutation group under the composition operation, namely the \emph{automorphism group} $\Aut(G)$.

In the following, we only consider uncolored graphs for the sake of simplicity. 
Let us remark, however, that we could use exactly the same machinery for colored graphs (see \cite{McKay201494}).

\paragraph{Refinement.} In the following, we want to \emph{individualize} vertices and \emph{refine} colorings. Individualizing vertices in a coloring is a process that artificially forces the vertex to form its own singleton cell. We use $\nu \in V^*$ to denote a sequence of vertices. We use such a sequence to record which vertices have been individualized. The expression~$\nu.v$ denotes the sequence~$\nu$ appended by~$v\in V$.

A \emph{refinement} is a function $\Refx \colon G \times V^* \to \Pi$. Given a graph $G$ and sequence of vertices $\nu$, it must satisfy the following properties:
\begin{itemize}
	\item It is invariant under isomorphism, i.e., $\Refx(G^\varphi, \nu^\varphi) = \Refx(G, \nu)^\varphi$ holds for all $\varphi \in S_n$.
	\item It respects vertices in $\nu$ as being individualized, i.e., $\{v\}$ is a singleton cell in $\Refx(G, \nu)$ for all $v \in \nu$.
\end{itemize}

In practice, variants of the 1-dimensional Weisfeiler-Leman algorithm (commonly referred to as \emph{color refinement}) are used as refinement procedures. Intuitively they classify vertices according to their degree and the degrees of their neighbors and the degrees of the neighbors of the neighbors and so on. The refinement as used throughout paper is summarized in Algorithm~\ref{alg:refine}.

The algorithm overapproximates the orbit partition by first coloring vertices using their degree. 
Then, this information is propagated iteratively through the graph, partitioning colors further by considering the colors of neighbors.
The algorithm can be implemented in quasi-linear time, i.e., $\mathcal{O}((n+m) \log n)$ where $m$ is the number of edges in the graph.

A crucial property is that the algorithm partitions vertices in an isomorphism-invariant manner: this implies that whenever two input graphs are isomorphic the resulting partitioning must be equivalent. 
Conversely, whenever refinement results in differing partitions, the provided colored graphs can not be isomorphic.  Let us point out that for isomorphism invariance Lines~\ref{line:refine:split} and~\ref{line:refine:largest} need to be implemented isomorphism invariantly.

\begin{algorithm2e}[t] \label{alg:refine}
	\SetAlgoLined
	\SetAlgoNoEnd
	\caption[Refinement Procedure]{Basic Color Refinement}
	\Fn{\Refine{G, $\pi$, $\nu$}}{
		\SetKwInOut{Input}{Input}
		\SetKwInOut{Output}{Output}
		\Input{graph $G$, coloring $\pi$, list of vertices $\nu$}
		\Output{refined coloring $\pi'$}
		initialize empty stack $W$\;
		push all cells of $\pi$ and $\nu$ onto $W$\;
		\While{{\normalfont $W$ is non-empty}}{
			pop a cell $C$ from $W$\;
			\For{{\normalfont each cell $X$ containing a neighbor of a vertex in $C$}}{
				for each vertex in $X$ count its neighbors in $C$ \;
				split $X$ into $X_1, \dots{}, X_k$ in $\pi$, according to neighbor counts\; \label{line:refine:split}
				let $X_i$ be one of the largest cells of $X_1, \dots{}, X_k$\; \label{line:refine:largest}
				push all sets $X_1, \dots{}, X_k$ except $X_i$  onto $W$\; \label{line:refine:putstack}
				\lIf{$X \in W$}{replace $X$ in $W$ with $X_i$}
			}
		}
		\Return{$\pi$}
	}
\end{algorithm2e}

\paragraph{Cell Selector.} If refinement classifies all vertices into different cells, determining automorphisms and isomorphisms for graphs is easy, since cells have to be preserved. 
Otherwise, individualization is used to artificially single out a vertex inside a non-singleton cell. 
The task of a \emph{cell selector} is to isomorphism invariantly pick a non-singleton cell of the coloring. 
In the individualization refinement paradigm, all vertices of the selected cell will then be individualized one after the other using some form of backtracking. After individualization, refinement is applied again.
Formally, a cell selector is a function $\Sel \colon G \times V^* \to 2^V$ satisfying the following properties: 
\begin{itemize}
	\item It is invariant under isomorphism, i.e., $\Sel(G,\pi^\varphi) = \Sel(G,\pi)^\varphi$ holds for all $\varphi \in S_n$.
	\item If $\pi$ is discrete then $\Sel(G,\pi) = \emptyset$.
	\item If $\pi$ is not discrete then  $|\Sel(G,\pi)| > 1$ and $\Sel(G,\pi)$ is a cell of $\pi$.
\end{itemize}

\paragraph{Search Tree.} With the functions $\Refx$ and $\Sel$ at hand, we are now ready to define the \emph{search tree}. For a graph $G$ we use $\mathcal{T}_{(\Refx, \Sel)}(G)$ to denote the search tree of~$G$ with respect to refinement operator~$\Refx$ and cell selector~$\Sel$. The search tree is constructed as follows: each node of the search tree corresponds to a sequence of vertices of~$G$.
\begin{itemize}
	\item The root of $\mathcal{T}_{(\Refx, \Sel)}(G)$ is the empty sequence $\epsilon$.
	\item If $\nu$ is a node in $\mathcal{T}_{(\Refx, \Sel)}(G)$ and $C = \Sel(G,\Refx(G, \nu))$ holds, then its children are $\{\nu.v \; | \; v \in C \}$.
\end{itemize}

\noindent With $\mathcal{T}_{(\Refx, \Sel)}(G, \nu)$ we denote the subtree of $\mathcal{T}_{(\Refx, \Sel)}(G)$ rooted in $\nu$. We may omit the indices $\Sel$ and $\Refx$ if they are apparent from context. Note that leaves of a tree correspond to discrete colorings of the graph, and therefore to permutations of $V$.

We recite the following crucial facts on isomorphism invariance of the search tree as given in \cite{McKay201494}, which follows from the isomorphism invariance of $\Sel$ and $\Refx$:

\begin{lemma} \label{lem:tree_invariant} For a graph $G$ and $\varphi \in S_n$ we have $\mathcal{T}(G)^\varphi = \mathcal{T}(G^\varphi)$. 
\end{lemma}

\begin{corollary} \label{lem:auto_tree_correspondence1} If $\nu$ is a node of $\mathcal{T}(G)$ and $\varphi \in \Aut(G)$, then $\nu^\varphi$ is a node of $\mathcal{T}(G)$ and $\mathcal{T}(G, \nu)^\varphi = \mathcal{T}(G, \nu^\varphi)$.
\end{corollary}

\noindent We have yet to mention how the search tree is used to find automorphisms and isomorphisms of graphs. For this we read off automorphisms and isomorphisms from the tree by looking at the colorings of leaves:

\begin{lemma} \label{lem:leaf_auto_correspondence2} If $\nu$ and $\nu'$ are leaves of $\mathcal{T}(G)$, then there exists an automorphism $\varphi \in \Aut(G)$ such that $\nu = \varphi(\nu')$, if and only if $\Refx(G, \nu')^{-1} \cdot \Refx(G, \nu)$ is an automorphism of $G$. 
\end{lemma}
\begin{proof} Set $\varphi' = \Refx(G, \nu')^{-1} \cdot \Refx(G, \nu)$, which is a well-defined permutation on $V$ since $\nu$ and $\nu'$ are leaves.
	
	If $\varphi$ is an automorphism with $\nu = \varphi(\nu')$, then $\Refx(G, \nu') \cdot \varphi = \Refx(G, \nu)$ holds. But then,
	\begin{align*}
		\varphi = \Refx(G, \nu')^{-1} \cdot \Refx(G, \nu') \cdot \varphi =\\ \Refx(G, \nu')^{-1} \cdot  \Refx(G, \nu) = \varphi',
	\end{align*}
	\noindent proving the first direction.
	
	Assume now $\varphi' \in \Aut(G)$. Then $\nu' = \varphi'(\nu)$ holds since $\nu$ and $\nu'$ are individualized in their respective coloring and must be mapped to each other. 
\end{proof}
Combining Lemma~\ref{lem:tree_invariant} and Lemma~\ref{lem:leaf_auto_correspondence2} shows that isomorphisms between graphs can be found similarly.
For a fixed $\varphi \in S_n$, we call $v' \in \mathcal{T}(G^{\varphi})$ an \emph{occurrence} of $v \in \mathcal{T}(G)$ in $\mathcal{T}(G^{\varphi})$ whenever there is an isomorphism~$\varphi'$ from~$G$ to~$G^{\varphi}$ for which~$\varphi'(v') = v$ holds.

\paragraph{Invariants.}
We define the notion of \emph{node invariants}. A node invariant $\Inv \colon G \times V^* \to I$ is a function mapping nodes of the tree to some totally ordered set $I$. We require some further properties:
\begin{itemize}
	\item The invariant must be invariant under isomorphism, i.e., we require $\Inv(G, \nu_1) = \Inv(G^\varphi, \nu_1^\varphi)$ for all $\varphi \in S_n$.
	\item If $|\nu_1| = |\nu_2|$ and $\Inv(G, \nu_1) < \Inv(G, \nu_2)$, then for all leaves $\nu_1' \in \mathcal{T}(G, \nu_1)$ and $\nu_2' \in \mathcal{T}(G, \nu_2)$ we require $\Inv(G, \nu_1') < \Inv(G, \nu_2')$.
\end{itemize}

\noindent For any invariant $\Inv$, the following holds:
\begin{lemma} \label{lem:invariant_pruning} Let $\nu, \nu'$ be leaves of $\mathcal{T}(G)$. If there is an automorphism $\varphi \in \Aut(G)$ such that $\nu = \varphi(\nu')$, then $\Inv(G, \nu) = \Inv(G, \nu')$ holds.
\end{lemma}
\begin{proof} This follows from the equalities~$\Inv(G, \nu) = \Inv(G^\varphi, \nu'^\varphi) = \Inv(G, \nu')$.
\end{proof} 
\section{Algorithmic Foundation} \label{sec:bidirectional}
In this section, we present the probabilistic bidirectional search algorithm for testing isomorphism of two graphs. 
The overall algorithm is based on the bidirectional traversal strategy presented in \cite{theorypaper}. In that paper traversal strategies for search algorithms in the context of symmetries are theoretically analyzed. However, this is done in an abstract search tree model. It is shown that some randomized search strategies asymptotically outperform the best possible deterministic traversal strategies. We describe now a realization of the strategy within the individualization-refinement framework. 

We then show that the algorithm exploits automorphisms without explicitly computing any part of the automorphism group. This crucially enables the algorithm to efficiently decide graph isomorphism without implicitly computing the automorphism group. We should note that the exploitation of automorphisms cannot be captured by the model presented in~\cite{theorypaper}. 

\subsection{No Isomorphism?~Probably} \label{sec:bidirectional:algorithm}

\begin{algorithm2e}[t] \label{alg:random_walk}
	\SetAlgoLined
	\SetAlgoNoEnd
	\caption[Random Walk of the Search Tree]{Random Walk in the Search Tree}
	\Fn{\RandomWalk{$G$}}{
		\SetKwInOut{Input}{Input}
		\SetKwInOut{Output}{Output}
		\Input{graph $G$}
		\Output{a random leaf of the search tree}
		$base$  := ()\;
		$c$  := \RRef{G, $[v \mapsto 1]$, base}\;
		$cell$     := \SSel{col}\;
		\While{c $\neq \emptyset$}{
			pick $v$ of $cell$ uniformly at random\;
			$base$ := $base.v$ \tcp*{append v to base}
			$c$ := \RRef{$G, c, base$}\;
			$cell$     := \SSel{$c$}\;
		}
		\Return{col}\;
	}
\end{algorithm2e}

The foundation of the probabilistic isomorphism test are \emph{random walks} in the underlying individualization-refinement search trees. 
The procedure is described in Algorithm~\ref{alg:random_walk}:
a random walk is performed by repeatedly refining and individualizing a random vertex of the target cell until the coloring becomes discrete, i.e., a leaf of the search tree is found.
This constitutes a random root-to-leaf walk of the individualization refinement search tree.

First, we need to make the following observation: assume we have two isomorphic graphs $G_1$, $G_2$. 
Now, due to isomorphism invariance, their respective search trees $\mathcal{T}(G_1)$ and $\mathcal{T}(G_2)$ are also isomorphic (Lemma~\ref{lem:tree_invariant}). 
Assume we fix some leaf $\tau \in \mathcal{T}(G_1)$ and try to find occurrences of it through random walks of the tree. 
Towards finding $\tau$, we always perform two random walks: one in $\mathcal{T}(G_1)$ and one in $\mathcal{T}(G_2)$. 
Using our assumption that the trees \emph{are isomorphic}, we can observe that finding an occurrence of $\tau$ in $\mathcal{T}(G_1)$ or in $\mathcal{T}(G_2)$ is \emph{equally likely}. Contrarily, if the trees \emph{are not isomorphic}, i.e., if the graphs are not isomorphic, we are \emph{only} able to find occurrences of $\tau$ in $\mathcal{T}(G_1)$.

Algorithm~\ref{alg:random_iso} describes the \emph{probabilistic bidirectional search}, which is based on this observation. 
The algorithm improves upon using just a single leaf $\tau$ by memorizing two sets of leaves $L_1$ and $L_2$ for comparison. 
If a leaf is discovered in $\mathcal{T}(G_j)$ that is not an occurrence of a previously found leaf, it is added to $L_j$ and is used for subsequent testing. 
Whenever a leaf is an occurrence of a previously found leaf, it either reveals an isomorphism or automorphism:
In the case where an isomorphism is unveiled, we are done and simply terminate returning the isomorphism. 
Otherwise, we have discovered an automorphism. 
After a certain number of automorphisms have been accumulated, the algorithm determines that the graphs are non-isomorphic within the given error bound. 
As discussed previously, if graphs were isomorphic, there is an equal probability to find automorphisms and isomorphisms. Hence, we are highly unlikely to uncover many automorphisms without also uncovering an isomorphism. Figure~\ref{fig:algorithm} illustrates Algorithm~\ref{alg:random_iso}.
Let us now formally prove its correctness:

\begin{figure}
	\centering
	\begin{tikzpicture}[scale=0.5]
		\node (r0) at ( 0.0,  0.0) {}; 
		\node (s0) at (-3.0, -4.0) {}; 
		\node (s1) at ( 3.0, -4.0) {}; 
		\node (si1) at (-2.0, -4.0) {};
		\node (si2) at (-1.5, -4.0) {};
		\node (si3) at ( 0, -4.0) {}; 
		\node (si4) at ( 5.5, -4.0) {}; 
		\node (si5) at ( 7.0, -4.0) {}; 
		\node (si6) at ( 9.0, -4.0) {}; 
		
		\fill[fill=gray!20,draw] (r0.center)--(s0.center)--(s1.center)--(r0.center);
		
		\node (T1) at  (-2, 0) {$\mathcal{T}(G_1)$};
		
		\node (r01) at ( 0.0 + 7,  0.0) {}; 
		\node (s01) at (-3.0 + 7, -4.0) {}; 
		\node (s11) at ( 3.0 + 7, -4.0) {}; 
		\fill[fill=gray!20,draw] (r01.center)--(s01.center)--(s11.center)--(r01.center);
		
		\node (T2) at  (7 + 2, 0) {$\mathcal{T}(G_2)$};
		
		\draw[color=black, line width=1.5pt,] (si1) to [out=290, in=250] (si3) node [label=below:{\tiny \hspace{0.1cm} auto?}] {}; 
		\draw[color=black, line width=1.5pt, ] (si1) to [out=290, in=250] (si4) node [label=below:{\tiny \hspace{0.2cm} iso?}] {}; 
		
		\draw[color=black, line width=1.5pt,densely dotted] (r0) to [out=250, in=90] (si1); 
		\draw[color=black, line width=1.5pt,densely dotted] (r0) to [out=270, in=90] (si2); 
		\draw[color=black, line width=1.5pt,densely dotted] (r0) to [out=290, in=90] (si3); 
		
		\draw[color=black, line width=1.5pt,densely dotted] (r01) to [out=250, in=90] (si4); 
		\draw[color=black, line width=1.5pt,densely dotted] (r01) to [out=270, in=90] (si5); 
		\draw[color=black, line width=1.5pt,densely dotted] (r01) to [out=290, in=90] (si6); 
		
		\draw[color=black, fill=white]  (r01) circle (.15);
		\draw[color=black, fill=white]  (r0) circle (.15);
		
		\draw[color=black, fill=orange]  (si1) circle (.15);
		\draw[color=black, fill=orange]  (si2) circle (.15);
		\draw[color=black, fill=orange]  (si3) circle (.15);
		\draw[color=black, fill=orange]  (si4) circle (.15);
		\draw[color=black, fill=orange]  (si5) circle (.15);
		\draw[color=black, fill=orange]  (si6) circle (.15);
	\end{tikzpicture}
	\caption{The probabilistic bidirectional search algorithm simultaneously samples uniform random leaves in both trees. It then tests for automorphisms within a tree and isomorphisms across trees to perform the probabilistic test.} \label{fig:algorithm}
\end{figure}
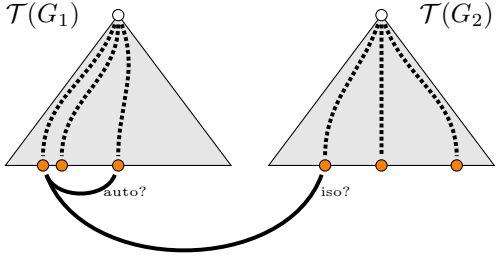

\begin{algorithm2e}[t!] \label{alg:random_iso}
	\SetAlgoLined
	\SetAlgoNoEnd
	\caption[Probabilistic Bidirectional Search]{Probabilistic Bidirectional Search in Individualization-Refinement}
	\Fn{\RandomIso{$G_1$, $G_2$, $\epsilon$}}{
		\SetKwInOut{Input}{Input}
		\SetKwInOut{Output}{Output}
		\Input{graphs $G_1, G_2$ and probability $\epsilon$}
		\Output{isomorphism $\varphi$ between the graphs with probability at least $1 - \epsilon$ if it exists, $\bot$ otherwise}
		$c$    := $0$\;
		$L_1 := L_2 := \emptyset$\;
		\While{$c \leq \ceil{-\log_2(\epsilon)}$}{
			$f_{(aut, 1)}$ := $f_{(aut, 2)}$ := $\mathtt{false}$\;
			$l_1$ := \RandomWalk{$G_1$}\;
			$l_2$ := \RandomWalk{$G_2$}\;
			\lIf{$l_1 \cdot l_2^{-1}(G_2) = G_1$}{\Return{$l_1 \cdot l_2^{-1}$}}
			\For{$i \in \{1, 2\}$}{ \label{alg:random_iso:search}
				\For{$l' \in L_{(3-i)}$}{
					$\varphi_{(iso, i)}$ := $l_i \cdot {l'}^{-1}$\;
					$\varphi_{(aut, 3-i)}$ := $l_{(3-i)} \cdot {l'}^{-1}$\;
					\lIf{$\varphi_{(iso, i)}(G_{3-i}) = G_{i}$}{
						\Return{$\varphi_{(iso, i)}$} 
					}
					\lIf{$\varphi_{(aut, 3-i)}(G_{(3-i)}) = G_{(3-i)}$}{
						$f_{(aut, 3-i)} := \mathtt{true}$
					}
				}
			}
			\lIf{$\neg f_{(aut, 1)}$}{$L_1 := L_1 \cup \{l_1\}$}
			\lIf{$\neg f_{(aut, 2)}$}{$L_2 := L_2 \cup \{l_2\}$}
			\lIf{$f_{(aut, 1)} \vee f_{(aut, 2)}$}{$c$ += $1$}
		}
		\Return{$\bot$}\;
	}
\end{algorithm2e}

\begin{lemma} Given graphs $G_1, G_2$ and probability $\epsilon$, Algorithm~\ref{alg:random_iso} produces an isomorphism $\varphi$ (such that $G_1 = \varphi(G_2)$) with probability at least $1 - \epsilon$ if it exists and returns $\bot$ otherwise.
\end{lemma}
\begin{proof} 
	First, observe that a discovered isomorphism is certified before being returned, which ensures that if the algorithm returns an isomorphism, it is always correct. The algorithm can therefore only fail by not finding an isomorphism despite its existence.

	We interpret the computation as a sequence of \emph{tests}. 
	A test performs random walks of the search trees until one automorphism or isomorphism is found. 
	Hence, it is a sequence of $i$ iterations: in each iteration $j < i$, neither $l_1$ nor $l_2$ uncover an isomorphism or automorphism. 
	The algorithm neither terminates, nor is $c$ incremented. 
	In iteration $i$, an automorphism or isomorphism is found. 
	Now, note that when $G_1$ or $G_2$ are isomorphic, a random leaf contained in $L_i$ can equally likely be found in $G_1$ or $G_2$. 
	Hence, finding an automorphism or isomorphism in a test is equally likely, i.e., the probability is $\frac{1}{2}$ for each outcome. 
	Anytime we find an automorphism but no isomorphism, we increment $c$ by $1$. We terminate when $c$ reaches $d$. 
	Assuming the graphs are isomorphic, the probability of this outcome is thus bounded by $(\frac{1}{2})^d$.
\end{proof}

\noindent The avid reader may remark that the algorithm as presented neither performs \emph{invariant pruning} nor \emph{automorphism pruning}, which are common practice in all state-of-the-art tools. 
However, both omissions are intentional.
In the next section, we provide a runtime analysis which shows that the algorithm already \emph{implicitly} prunes using automorphisms: it becomes proportionally faster in the presence of automorphisms.
Furthermore, the algorithm has a sub-linear worst-case runtime in the size of the search tree, which depends on explicitly \emph{not} always using invariant pruning. Regarding invariant pruning,
Section~\ref{sec:practical:deviation} introduces an adapted way of applying invariants in the probabilistic setting.
In Section~\ref{sec:benchmarks}, benchmarks then show that traditional invariant pruning as performed by deterministic solvers is rarely required for the probabilistic approach.  

\subsection{No Automorphism Pruning? Yes}
 \label{sec:bidirectional:runtime}
A crucial point we want to make is that the algorithm as presented performs perfect automorphism pruning.

Let us first discuss how automorphism pruning is performed in other IR-algorithms. 
In these, automorphisms are usually discovered by finding multiple occurrences of leaves. 
There is a close relationship between the automorphism group size and the number of occurrences of a particular leaf.

\begin{lemma} \label{lem:num_occurrences}
	Let $G$ be a graph. For every leaf $l \in \mathcal{T}(G)$, there exist $|\Aut(G)|$ occurrences of $l$.
\end{lemma}
\begin{proof}
	Let $l$ be a leaf of $\mathcal{T}(G)$. 
	Due to isomorphism invariance (Lemma~\ref{lem:tree_invariant}), applying a non-trivial automorphism $\varphi \in \Aut(G)\setminus\{\id\}$ yields a distinct leaf $\varphi(l) \neq l$. 
	This accounts for $|\Aut(G)|$ occurrences in total. 
\end{proof}

Suppose we have discovered a number of automorphisms. We compute the subgroup~$\Gamma$ generated by these. 
It is then possible to prune the search tree, essentially constructing the quotient~$\mathcal{T}/\Gamma$. This quotient is the graph defined on the orbits of the vertices of~$\mathcal{T}$ under~$\Gamma$. Two orbits are adjacent if they contain adjacent vertices.
Intuitively, if some node of the search tree is mapped to another node using one of the automorphisms then we can remove one of the nodes.
The best we could hope for here is to discover all automorphism cheaply to be able to prune as much of the search tree as possible. In that case we say that automorphism pruning has been applied exhaustively.
The following lemma describes the size of a search tree after automorphism pruning has been applied exhaustively.

\begin{lemma} \label{lem:full_auto}
	Let $G$ be a graph. 
	The quotient tree~$\mathcal{T}/\Aut(G)$ of the search tree modulo the automorphism group has 
	$|L(\mathcal{T}/\Aut(G))| = \frac{|L(\mathcal{T})|}{|\Aut(G)|}$ leaves.
\end{lemma}
\begin{proof} 
	Note that a leaf in the quotient search tree~$\mathcal{T}/\Aut(G)$ is an equivalence class of leaves of the original search tree.
	Let $l$ be a leaf of $\mathcal{T}$. From Lemma~\ref{lem:num_occurrences} it follows that there are $|\Aut(G)| - 1$ other occurrences of $l$ equivalent to~$l$ under~$|\Aut(G)|$.
\end{proof}

Now, we show that Algorithm~\ref{alg:random_iso} implicitly exploits automorphisms without ever having to handle them explicitly.
Towards this goal we analyze its runtime.
Since termination in the algorithm depends on randomized events we consider expected runtime.  
In the implementation, a single path in the tree can be calculated in time $\mathcal{O}((n+m) \log n)$, where $m$ is the number of edges in the considered graph. 
In Section~\ref{sec:practical} we discuss how in practice comparing a new leaf to previously found leaves can be handled efficiently through hashing (instead of the linear search following Line~\ref{alg:random_iso:search}). 
Hence, we assume that this can be done in $\mathcal{O}(1)$.  
The only unknown is therefore the number of nodes visited in the search tree.

Overall, as typical for IR-type algorithms, we can therefore measure the runtime in the number of nodes visited in the search trees. 
Specifically, 
we describe the runtime of the algorithm on two graphs $G_1, G_2$ in terms of the sizes of the search trees $|\mathcal{T}(G_1)|, |\mathcal{T}(G_2)|$, their heights~$h(\mathcal{T}(G_1)), h(\mathcal{T}(G_2))$ and the desired error probability.
\begin{lemma} \label{lem:runtime} 
	Let $G_1, G_2$ be graphs.
	In the worst case, Algorithm~\ref{alg:random_iso} then visits in expectation 
	\begin{align*}
		\mathcal{O}(\ceil{-\log_2(\epsilon)} \cdot \max\{h(\mathcal{T}(G_1)), h(\mathcal{T}(G_2))\} \cdot \\ \min\{\sqrt{\frac{|\mathcal{T}(G_1)|}{|\Aut(G_1)|}}, \sqrt{\frac{|\mathcal{T}(G_2)|}{|\Aut(G_2)|}}\})
	\end{align*}
	nodes of the search tree.
\end{lemma}
\begin{proof}
This proof is similar to the proof of~\cite[Lemma 2]{theorypaper}, but there automorphisms are not part of the consideration.
To prove the claim, we calculate the expected number of leaves before termination. Note that we may consider the number of leaves instead of nodes by adding the multiplicative factor $\max\{h(\mathcal{T}(G_1)), h(\mathcal{T}(G_2))\}$ for the height of the search trees to our runtime. Furthermore, we assume that the algorithm terminates due to reaching the condition $c = \ceil{-\log_2(\epsilon)}$. This suffices to give an upper bound: earlier termination due to finding isomorphisms can only lead to a smaller expected number of leaves.

Let us now calculate the expected number of leaves explored before the first discovery of an automorphism. We assume that in the $i$-th iteration $L_1$ and $L_2$ each contain at least $i$ leaves.
Otherwise, a previous iteration already uncovered an automorphism or isomorphism: hence, the assumption suffices for a lower bound of the probability. Furthermore, we assume that the probability to find a leaf is uniform across all leaves: if probabilities are non-uniform, the chance for finding some leaf twice strictly increases.
The probability of finding an automorphism in $L_j$ (with $j \in \{1, 2\}$) within $i$ iterations is therefore at least $\frac{i}{|\mathcal{T}(G_j)|}$.

We now argue that the likelihood of finding an occurrence in the search tree through random walks is amplified by the size of the automorphism group. 
Let $G$ be a graph and $\tau \in L(\mathcal{T}(G))$.
In $\mathcal{T}(G)$, there are $|\Aut(G)|$ occurrences of $\tau$ (Lemma~\ref{lem:num_occurrences}). 
Let $p$ be the probability of finding the node $\tau$ through a random walk of $\mathcal{T}(G)$.
But due to isomorphism invariance of $\mathcal{T}(G)$ (Lemma~\ref{lem:tree_invariant}), the probability of finding a specific occurrence $\tau'$ of $\tau$ is also $p$. Hence, the probability to find \emph{any} occurrence is $|\Aut(G)|p$. 
In our specific case, the probability of finding an automorphism in $L_j$ (with $j \in \{1, 2\}$) within $i$ iterations is therefore at least $\frac{|\Aut(G_j)|\cdot i}{|\mathcal{T}(G_j)|}$. 

Consider running $2\sqrt{\frac{|\mathcal{T}(G_j)|}{|\Aut(G_j)|}}$ iterations of the algorithm. After $\sqrt{\frac{|\mathcal{T}(G_j)|}{|\Aut(G_j)|}}$ iterations, the probability for finding an automorphism in $\mathcal{T}(G_j)$ is at least $\frac{1}{\nicefrac{|\mathcal{T}(G_j)|}{|\Aut(G_j)|}}$.
This suffices to show that in expectation, the algorithm finds an automorphism after $2\sqrt{\frac{|\mathcal{T}(G_j)|}{|\Aut(G_j)|}}$ iterations. 
Repeating the argument $\ceil{-\log_2(\epsilon)}$ many times (to find all the necessary automorphisms for termination) shows the claimed runtime.
\end{proof}

\noindent Lemma~\ref{lem:runtime} shows that search trees are implicitly pruned using automorphisms: isomorphic copies of leaves actively contribute towards termination. 
In particular, in conjunction with Lemma~\ref{lem:full_auto}, we can see that the algorithm exploits \emph{all} automorphisms.
Philosophically, one can think of the random walks being performed on the quotient tree, however, one has to be aware that the sampling of children in the selected cell is not uniform.

\section{Engineering} \label{sec:practical}

The algorithm is implemented in \texttt{C++} and the implementation is called \dejavu{}.
The implementation would freely be available at \cite{webpage}. 

The implementation of our solver follows Algorithm~\ref{alg:random_iso} closely. It uses highly-engineered versions of the subroutines $\Refx$, $\Inv$ and $\Sel$. Their implementation is based on the algorithms from \cite{McKay201494} 
and in part even reverse-engineered from the source code of \Traces{}. 

To summarize, the refinement used is a version of the basic color refinement routine (see Algorithm~\ref{alg:refine}). 
Following the implementation of \Traces{}, it features several versions of the algorithm designed for different densities of graphs, among other optimizations. 
The cell selector $\Sel$ always picks the first largest cell of the coloring. 
We use a caching strategy to speedup the selection process.
We use different invariants depending on the specific use case, described below.

We now present further optimizations made in the implementation of Algorithm~\ref{alg:random_iso}. 
\subsection{Leaf Storage} \label{sec:practical:storage}
When storing and comparing leaves, we use a hash map in conjunction with an invariant. The invariant is analogous to the ones described in~\cite{McKay201494}, blending together most of the isomorphism-invariant information of a leaf into a single value. 

Leaves are then stored in a hashmap, using the invariant as the key value. Since we only store non-isomorphic leaves, isomorphic leaves do not produce collisions in the map. Collisions in this map for non-isomorphic leaves, empirically, are also very rare. The reason would be that the used invariant is almost a complete invariant. 

Storing the entire partition corresponding to a leaf (amounting to storing~$n$ numbers) for all of the explored leaves, quickly becomes very memory-intensive for many graphs. Therefore, the implementation only stores entire leaves up to a predetermined memory limit and then uses a cheaper method as follows. Instead of storing the entire leaf, only the individualized vertices are stored, i.e., the path taken through the tree. If the solver tries to derive automorphisms or isomorphisms from leaves at a later point, the path is taken again to recompute the coloring. Since the total number of leaves ever used for automorphism or isomorphism derivation is small (specifically at most $\ceil{-\log_2(\epsilon)}$), this cost is quickly amortized. For instances for which not many paths are computed overall, it is however highly beneficial to store the first few leaves in their entirety to prevent recomputation.

In our tests, this method was sufficient to conserve memory usage and the solver was never limited by memory (but rather time).

\subsection[k-deviation Trees]{$k$-deviation Trees} \label{sec:practical:deviation}
In the implementation, we add an additional step to the algorithm. 
Before performing the probabilistic bidirectional search as described in Algorithm~\ref{alg:random_iso}, we essentially perform the same algorithm but on a pruned tree. 
This helps the solver detect ``clearly non-isomorphic'' search trees more efficiently.

 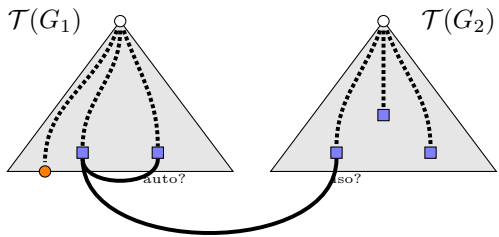
\begin{figure}
	\centering
	\begin{tikzpicture}[scale=0.5]
		\node (r0) at ( 0.0,  0.0) {}; 
		\node (s0) at (-3.0, -4.0) {}; 
		\node (s1) at ( 3.0, -4.0) {}; 
		\node (si1) at (-2.0, -4.0) {};
		
		\fill[fill=gray!20,draw] (r0.center)--(s0.center)--(s1.center)--(r0.center);
		
		\node (T1) at  (-2, 0) {$\mathcal{T}(G_1)$};
		
		\node (r01) at ( 0.0 + 7,  0.0) {}; 
		\node (s01) at (-3.0 + 7, -4.0) {}; 
		\node (s11) at ( 3.0 + 7, -4.0) {}; 
		\fill[fill=gray!20,draw] (r01.center)--(s01.center)--(s11.center)--(r01.center);
		
		\node (T2) at  (7 + 2, 0) {$\mathcal{T}(G_2)$};

		\node (si2) at (-1.0, -3.5) [draw,color=black,fill=lightblue,minimum size=1.5mm,inner sep=0pt] {};
		\node (si3) at ( 1.0, -3.5) [draw,color=black,fill=lightblue,minimum size=1.5mm,inner sep=0pt] {}; 
		\node (si4) at ( 5.75, -3.5) [draw,color=black,fill=lightblue,minimum size=1.5mm,inner sep=0pt] {}; 
		\node (si5) at ( 7.0, -2.5) [draw,color=black,fill=lightblue,minimum size=1.5mm,inner sep=0pt] {}; 
		\node (si6) at ( 8.25, -3.5) [draw,color=black,fill=lightblue,minimum size=1.5mm,inner sep=0pt] {}; 

		\draw[color=black, line width=1.5pt,] (si2) to [out=270, in=270] (si3) node [label=below:{\tiny \hspace{0.1cm} auto?}] {}; 
		\draw[color=black, line width=1.5pt, ] (si2) to [out=270, in=270] (si4) node [label=below:{\tiny \hspace{0.2cm} iso?}] {}; 
		
		\draw[color=black, line width=1.5pt,densely dotted] (r0) to [out=250, in=90] (si1); 
		\draw[color=black, line width=1.5pt,densely dotted] (r0) to [out=270, in=90] (si2); 
		\draw[color=black, line width=1.5pt,densely dotted] (r0) to [out=290, in=90] (si3); 
		
		\draw[color=black, line width=1.5pt,densely dotted] (r01) to [out=250, in=90] (si4); 
		\draw[color=black, line width=1.5pt,densely dotted] (r01) to [out=270, in=90] (si5); 
		\draw[color=black, line width=1.5pt,densely dotted] (r01) to [out=290, in=90] (si6); 
		
		\draw[color=black, fill=white]  (r01) circle (.15);
		\draw[color=black, fill=white]  (r0) circle (.15);
		
		\draw[color=black, fill=orange]  (si1) circle (.15);
	\end{tikzpicture}
	\caption{Probing on a pruned tree continuously compares the invariant to that of a target leaf $\tau$ (circular node). Once a deviation occurs (rectangular nodes), a (fake) leaf is established and the value of the deviation is compared to that of other fake leaves. Once the same deviation has occurred in both trees (a potential isomorphism), we switch to executing Algorithm~\ref{alg:random_iso} on the unpruned tree} \label{fig:algorithmdev}
\end{figure}

To do so, we first compute an arbitrary leaf in one of the trees, which we call the \emph{target leaf} $\tau$.
We then compute an invariant for $\tau$. 
In particular, we record a \emph{trace} invariant (as introduced by \Traces{}) while walking down the tree to $\tau$. 
The trace invariant records most of the isomorphism-invariant information \emph{during} the refinement steps, i.e., a trace of the computation is made. 
Essentially, the trace records the contents of the worklist $W$ of Algorithm~\ref{alg:refine}.
In subsequent walks, we then also record and compare this information while computing Algorithm~\ref{alg:refine}.
This enables a potential early-out: once the invariant of a subsequent walk deviates from the information recorded for $\tau$, we cannot find an occurrence of $\tau$. 

We now discuss how this can be exploited.
To describe the technique, we first define another node invariant, which we call the \emph{deviation value} $\Dev_{\Inv}\colon V^* \to \mathbb{N}^2 \cup \{\bot\}$. 
Consider a fixed trace $\Inv(\tau)$, which for our purposes will be the trace invariant of the target leaf $\tau$. 
The deviation value $\Dev_{\Inv}(\nu)$ for a node $\nu$ of the search tree is then defined as the pair~$(i,j)$ consisting of the position~$i$ where the traces first deviate and the corresponding value~$j$ in the trace $\Inv(\nu)$ that is different from $\Inv(\tau)$. 
If there are no differences, we set the deviation value to $\bot$ denoting ``no deviation''. 
Since the deviation value is a function of the invariant computed up until an isomorphism-invariant point, it is also naturally invariant under isomorphism.

Using these invariants, we then perform a variation of Algorithm~\ref{alg:random_iso}: as usual, we perform random walks.
However, as described above, we continuously record a trace invariant and compare it to $\Inv(\tau)$. 
Assume we are currently at node $\nu$ of the random walk.
If $\nu$ is not a leaf and the invariant does not deviate, i.e., $\Dev_{\Inv}(\nu) = \bot$, we continue walking down the tree.
If the invariant deviates, we stop Algorithm~\ref{alg:refine} early and record $\Dev_{\Inv}(\nu)$ as a (fake) leaf of the tree. 
Otherwise, we continue until the coloring becomes discrete and we reach an \emph{actual leaf}, i.e., a leaf of the underlying unpruned individualization-refinement search tree.

Note that the sets of leaves $L_i$ of Algorithm~\ref{alg:random_iso} can now contain, in addition to actual leaves of the search trees, inner nodes that deviated from $\tau$ (see Figure~\ref{fig:algorithmdev}).
If the algorithm finds automorphisms or isomorphisms of actual leaves we can proceed as usual.
In the case when the alleged automorphisms or isomorphisms only refer to deviations in inner nodes (fake leaves), we define the following behavior:
once $d(\epsilon)$ consecutive deviations occur within $\mathcal{T}_1$ or $\mathcal{T}_2$ (``automorphism'' of inner nodes), we conclude that graphs are non-isomorphic. 
If we ever find the same deviation in both trees (``isomorphism'' of inner nodes), we abort probing in the pruned tree and switch to Algorithm~\ref{alg:random_iso} on the unpruned tree.

While deviation values could be used precisely as described, in the implementation we use a slight but crucial variation. To make deviation values more distinct, it is sometimes beneficial to not use the early-out immediately. 
Instead, for a fixed constant $k$, color refinement is continued past the deviation for $k$ more cells of the worklist $W$, accumulating more information for the deviation value. 
The trade-off is as follows: if $k$ becomes larger, the early-out in color refinement is taken later, but deviation values become more distinct.

For a more global perspective note that we can actually apply Algorithm~\ref{alg:random_iso} on any isomorphism-invariant structure (in the technique just described an invariant subtree of the search tree). 
Specifically, here we try to sample on a cheaper and smaller structure that is however not always expressive enough to solve the problem effectively. By adjusting $k$ we shift the balance of this trade-off: if $k$ increases, so does the expressiveness and the cost of sampling.

\subsection{Blueprints} \label{sec:practical:blueprint}
Using an additional idea the technique described in the previous section can be performed more efficiently.

We introduce the concept of using the trace of the target leaf $\tau$ as a \emph{blueprint} for subsequent branches. 
When Algorithm~\ref{alg:refine} is computed, the trace records --- among other information --- all the information of the worklist $W$. 
Usually, for the deviation trees, this information is then checked for equality: if branches are isomorphic, this information must be equal due to isomorphism-invariance. 
However, we can also turn this observation around and treat the trace as a blueprint: if we assume branches are isomorphic, the trace already gives us the entire future worklist for a branch. 
Now, the way we exploit this idea is that when the trace is recorded, information is added as to whether a cell in $W$ was \emph{splitting} or \emph{non-splitting}. 
We consider a cell of the worklist splitting if Line~\ref{line:refine:split} produced new cells in the coloring for any connected cell. 
Essentially, we record whether the coloring was changed by $C$. 
If not, i.e., $C$ was non-splitting, we skip $C$ in the worklist of subsequent branches through the recorded information. 
Since $C$ did not manipulate the coloring we are guaranteed to get the correct stable coloring for isomorphic branches. 

Note that while $\tau$ is not chosen isomorphism-invariantly, it is fixed first and then used for both trees: hence any skipping of cells is actually performed in an isomorphism-invariant manner.
For non-isomorphic branches, the technique could on paper make the trace invariant weaker, but we did not observe any negative impact in our tests. 

\section{Benchmarks} \label{sec:benchmarks}
We compare running times with \nauty{}, \Traces{} and \conauto{}. \nauty{} and \Traces{} do not feature an ``isomorphism'' mode. 
Hence, we use canonical labeling on both input graphs and then compare the labeled graphs (as suggested in \cite{McKay:userguide}). 
We did preliminary testing using the automorphism mode and the disjoint union of graphs (as described in Section~\ref{sec:introduction}), which failed to be competitive with canonical labeling for most graph classes.
Furthermore, we compare with \conauto{}, a tool that does feature direct isomorphism testing. 
\conauto{} is however limited to a dense graph representation (adjacency matrices), implying that large graphs can often not be solved by it at all due to memory constraints. 
Therefore, we could not run \conauto{} on any graphs beyond order $3000$, which is why data points are missing.

For the sake of clarity, we omit the inclusion of the canonical labelling tool \bliss{} \cite{JunttilaKaski:ALENEX2007, DBLP:conf/tapas/JunttilaK11}. 
The results would however be similar to \nauty{} or \Traces{}:
we refer to \cite{McKay201494} for a comparison of \bliss{} to \nauty{}, \Traces{} and \conauto{}.

All benchmarks were run on an Intel Core i7 9700K processor with 16GB of RAM and Ubuntu 19.10. 
Error probability for \dejavu{} was set to below $1$\%, but since the error is one-sided there can be no erroneous results on non-isomorphic instances. 
The timeout is $60$ seconds. 
All results are given in milliseconds.

Whenever interesting non-isomorphic instances are available for all graph orders within a set, we give results for both isomorphic and non-isomorphic instances. 
Otherwise, we just give results for isomorphic instances. 
Overall, the results show, with few exceptions, clear overall improvements over state-of-the-art solvers.

We highlight several graph classes on which the running times show interesting behavior.

A very interesting case arises for random regular graphs (see Figure~\ref{fig:bench:ranreg}).
These graphs with $n$ nodes, say, result in search trees that almost surely have $n$ leaves immediately attached to the root.
Furthermore, they are asymmetric.
For the sake of argument, assume that color refinement runs in $\mathcal{O}(n)$ on these graphs. For the deterministic solvers, this in turn results in quadratic runtime, while \dejavu{} runs in $\mathcal{O}(n\sqrt{n})$.
\Traces{} however has a special strategy which is very effective on this set, namely the trace invariant. 
This enables it to abort computation for most of the leaves very early, resulting in quite modest quadratic runtime. 
In particular, it is still able to outperform \dejavu{} on the isomorphic instances of this benchmark set.
On the non-isomorphic instances, \dejavu{} does however also exploit the trace invariant through $k$-deviation trees.
This results in better runtimes on the set and we can clearly observe the asymptotic advantage, even over \Traces{}.

In Figure~\ref{fig:bench:multikef}, we can see that \dejavu{} is clearly outperformed by \conauto{} and \Traces{} on Kronecker eye flip graphs. 
By analyzing the search trees of these graphs, it becomes apparent that the \emph{invariant pruning} heuristic is very effective on these graphs: it is possible to cut off large parts of the search trees close to the root.
\dejavu{} currently does not exploit invariant pruning in this case
leading to inferior performance.
The surprising fact here is that this seems to be an exception: invariant pruning, which is a crucial tool of traditional deterministic solvers, does not seem to be required at all in the majority of benchmark sets for our probabilistic isomorphism test.

The remedy for graphs such as the Kronecker eye flip graphs is therefore obvious: by blending in some breadth-first traversal with invariant pruning (such as done in \Traces{}), we could boost performance on these graphs greatly. This remains as future work. 

\newcommand{\plotScale}[0]{0.525}

\begin{figure}
	\centering \includegraphics{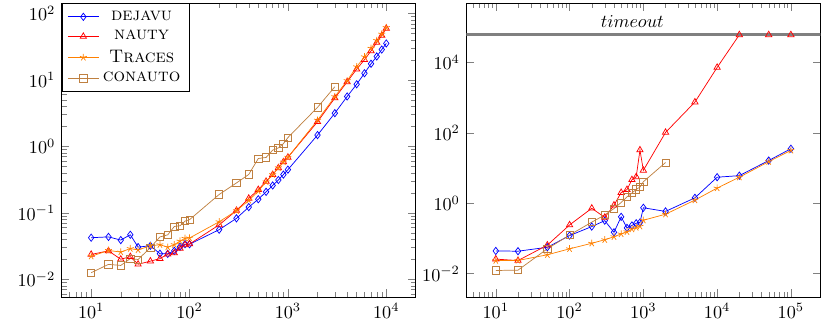}
	\caption{Isomorphic pairs random graphs with $\frac{1}{10}$ edge probability (left) and random trees (right).} \label{fig:ran}
\end{figure}

\begin{figure}
	\centering \includegraphics{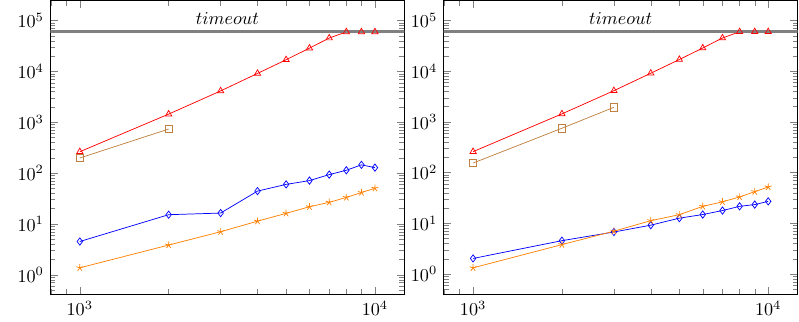}
	\caption{Isomorphic pairs random 3-regular graphs (left) and non-isomorphic pairs (right).} \label{fig:bench:ranreg}
\end{figure}

\begin{figure}
	\centering \includegraphics{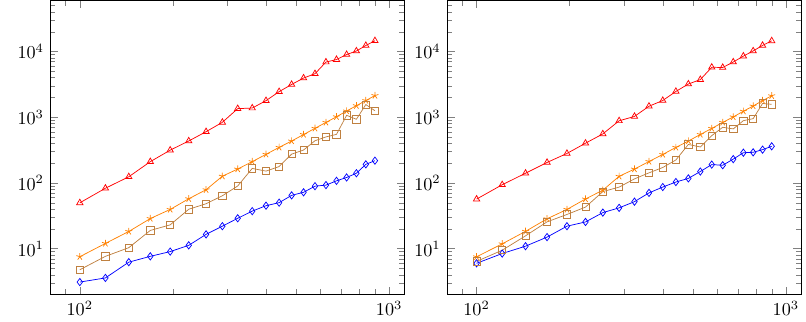}
	\caption{Isomorphic pairs of Latin square graphs with switched edges (left) and non-isomorphic pairs (right).} \label{fig:bench:latin}
\end{figure}

\begin{figure}
	\centering \includegraphics{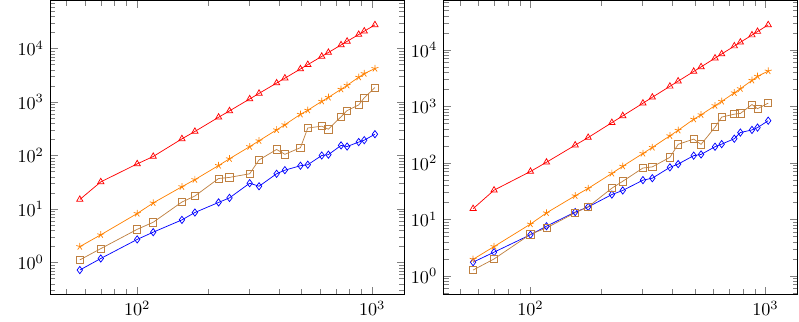}
	\caption{Isomorphic pairs of Steiner triple system graphs with switched edges (left) and non-isomorphic pairs (right).} \label{fig:bench:steiner}
\end{figure}

\begin{figure}
	\centering \includegraphics{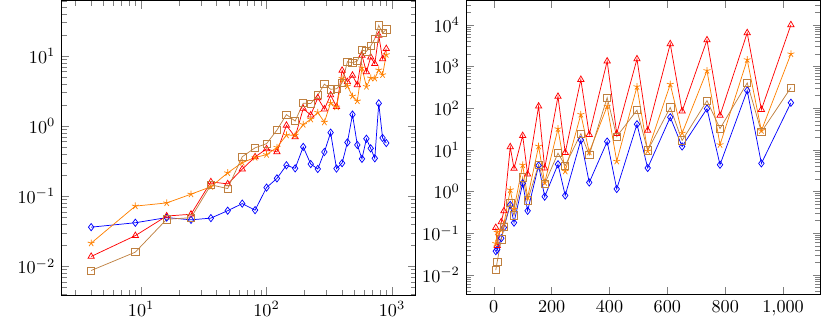}
	\caption{Isomorphic pairs of Latin square graphs (left) and Steiner triple systems (right).} \label{fig:bench:latinsteiner}
\end{figure}

\begin{figure}
	\centering \includegraphics{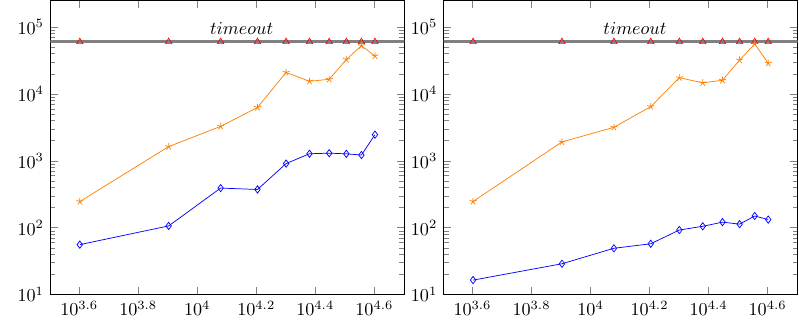}
	\caption{Isomorphic pairs of CFI graphs (left) and non-isomorphic pairs (right).} \label{fig:bench:cfi}
\end{figure}

\begin{figure}
	\centering \includegraphics{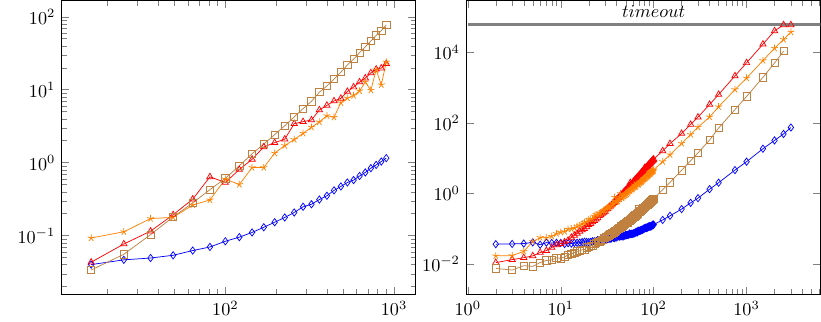}
	\caption{Isomorphic pairs of lattice graphs (left) and complete graphs (right).} \label{fig:bench:easysym}
\end{figure}

\begin{figure}
	\centering \includegraphics{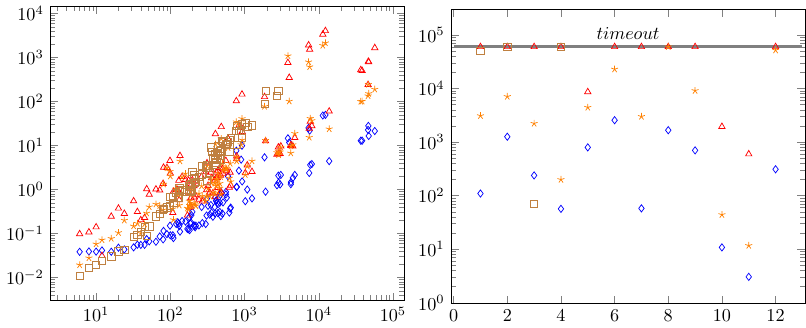}
	\caption{Isomorphic pairs of vertex transitive graphs (left) and combinatorial graphs (right).} \label{fig:bench:trancomb}
\end{figure}

\begin{figure}
	\centering \includegraphics{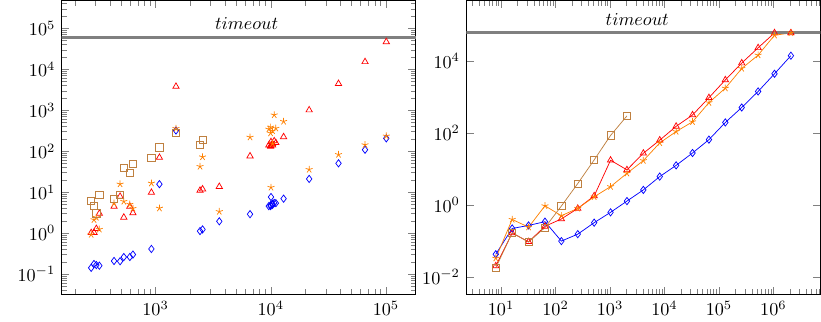}
	\caption{Isomorphic pairs of SAT formula model graphs (left) and hypercubes (right).} \label{fig:bench:dachyper}
\end{figure}

\begin{figure}
	\centering \includegraphics{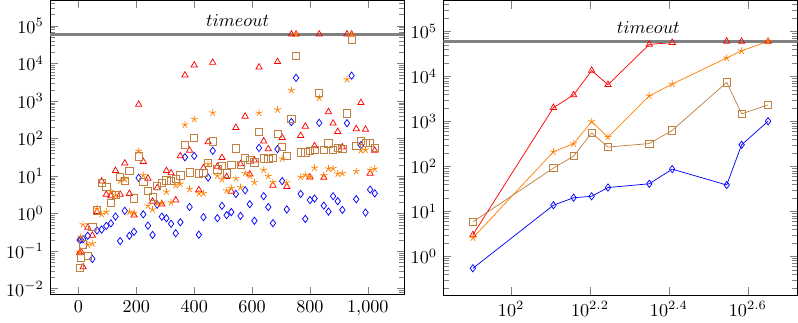}
	\caption{Isomorphic pairs of Hadamard matrices (left) and Hadamard matrices with switched edges (right).} \label{fig:bench:had}
\end{figure}

\begin{figure}
	\centering \includegraphics{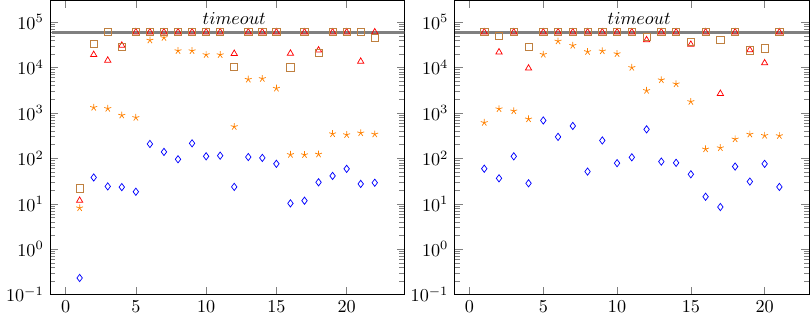}
	\caption{Isomorphic pairs of uncolored projective planes of order 16 (left) and (mostly) non-isomorphic pairs (right).} \label{fig:bench:pp16}
\end{figure}

\begin{figure}
	\centering \includegraphics{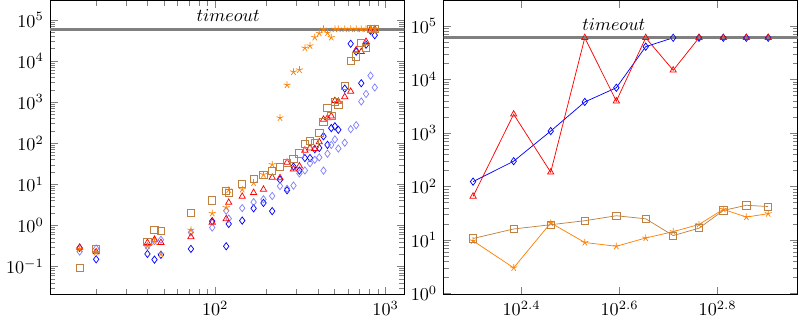}
	\caption{Isomorphic pairs of shrunken multipedes (left) and Kronecker eye flip graphs (right). The light blue plot on the left graph shows \dejavu{} with the \emph{smallest} cell selector, which is more comparable to the selector used by \conauto{} and \nauty{}.} \label{fig:bench:multikef}
\end{figure}

\section{Conclusions and Future Work} \label{sec:conclusions}
We have designed an isomorphism test that is faster than other state-of-the-art solutions for many graph classes. 
By using a probabilistic approach, the algorithm runs in time sub-linear in the size of the search tree. 
Furthermore, it completely avoids collecting symmetries, a novelty among state-of-the-art solvers.

While the experimental data clearly demonstrate superior performance of \dejavu{} over the other solvers on a wide variety of graphs, we think that this mostly shows that canonical labeling solvers are not the right tool for the job of isolated isomorphism testing. 
Our solver clearly gains advantages by \emph{not} computing canonical labelings and automorphism groups, while other solvers do so.
Overall this suggests that with the currently available algorithmic approaches, in practice, isomorphism testing is substantially easier than automorphism group computation and canonical labeling.

Future improvements to the tool can include blending breadth-first search with probabilistic bidirectional search to enable invariant pruning on the first few levels of the search tree. 

Regarding errors due to the randomized computation, a Las Vegas approach described in \cite{theorypaper} could eliminate errors completely. In fact that approach satisfies similar worst case bounds as the Monte-Carlo algorithm on which our algorithm is based.
The Las Vegas algorithm does however not seem to inherit the implicit automorphism group exploitation enjoyed by the Monte Carlo approach explored in this paper.  

\newpage
\bibliography{main}
\bibliographystyle{plain}

\end{document}